\documentclass{article}

\usepackage{amsmath}
\usepackage{amsfonts}

\usepackage{graphicx}

\newtheorem{theorem}{Theorem}[section]
\newtheorem{lemma}[theorem]{Lemma}
\newtheorem{corollary}[theorem]{Corollary}
\newtheorem{proposition}[theorem]{Proposition}

\newenvironment{proof}{\noindent\textsc{Proof: }}{\hspace{\stretch{1}}$\square$\medskip}

\begin{document}


\title
  {Upper bounds on the Witten index for supersymmetric
   lattice models by\\ discrete Morse theory}

\author{Alexander Engstr\"om\\
Department of Mathematics\\
Royal Institute of Technology\\
S-100 44 Stockholm\\
Sweden\\
\\
{\tt alexe@math.kth.se}
}

\date\today

\maketitle


%

\begin{abstract}
The Witten index for certain supersymmetric lattice models
treated by de Boer, van Eerten, Fendley, and Schoutens, can be formulated as a
topological invariant of simplicial complexes arising as
independence complexes of graphs. We prove a general
theorem on independence complexes using discrete Morse theory: 
If $G$ is a graph and $D$ a subset of its vertex set such that
$G\setminus D$ is a forest, then 
$ \sum_i \dim \tilde {H}_i({\tt Ind}(G);\mathbb{Q}) \leq |{\tt Ind}(G[D])| $.
We use the theorem to calculate upper bounds on
the Witten index for several classes of lattices.
These bounds confirm some of the computer calculations by van Eerten on
small lattices.

The cohomological method and the 3-rule of Fendley et al. is a special
case of when $G\setminus D$ lacks edges. We prove a generalized 3-rule
and introduce lattices in arbitrary dimensions satisfying it.
\end{abstract}

\section{Introduction}

This paper is motivated by combinatorial questions that arise in statistical
physics. To deal with the problems we use a discrete version of
Morse theory and algebraic topology. This short 
introduction to certain supersymmetric lattice models follows the work of
de Boer, van Eerten, Fendley, and Schoutens \cite{vE,FSB,FSE}
closely and we refer to them for the big picture. A lattice is
a graph, and vertices can be occupied by certain elementary 
particles called fermions. But two fermions are not allowed to occupy
adjacent vertices. 
The Witten index $W=\textrm{tr}((-1)^Fe^{-\beta H})$ for the
Hamiltonian $H$ turns out to be independent of $\beta$, and
in the limit $\beta\rightarrow 0$ it is 
\[f_0-f_1+f_2-f_3+\cdots, \]
where $f_i$ is the number of ways $i$ fermions can be distributed 
on the lattice. As exemplified in \cite{FSB}, if the lattice
is a cube then $W=1-8+16-8+2=3$. The Witten index is used to
estimate the number of ground states of a system.

There is a beautiful connection between combinatorial topology
and physics
first used by Jonsson \cite{J2} to prove two conjectures
from \cite{FSE} and later explored by Bousquet-M\'elou, 
Linusson, and Nevo \cite{LNM}. For a simplicial complex
with $f_i$ faces of dimension $i-1$ the reduced Euler
characteristic is $-f_0+f_1-f_2+f_3-\cdots,$ which
is $-W$. The simplicial complex of allowed
fermion configurations on a graph is usually called the
independence complex of the graph, and our
main result, Theorem~\ref{theorem:main}, is a tool
for bounding expressions like the reduced Euler
characteristic (and hence the Witten index).
In section~\ref{sec:bounds} we apply our estimation technique
on lattice types for which van Eerten \cite{vE} approximated the 
Witten index using transfer matrices that could
be computer treated. 

In the last section we generalize
the cohomological method and the 3-rule of 
Fendley, Halverson, Huijse, and Schoutens \cite{FHHS, FS, HS}.
We present lattices of any dimension that satisfy
the generalized 3-rule and give good lower bounds on
their number of ground states.

\subsection{Notation}

We will use $2^V$ to denote the 
set of all subsets of $V$.
An \emph{abstract simplicial complex} $\Sigma$ with 
vertex set $\Sigma^0$ is a subset of $2^{\Sigma^0}$
satisfying
$ \sigma\subseteq\tau\in\Sigma\Rightarrow \sigma\in\Sigma. $
We will often patch together simplicial complexes combinatorially
and in that case it is useful to allow $\emptyset\in\Sigma$.
All graphs and simplicial complexes in this paper are finite.
The \emph{face poset} $\mathcal{F}(\Sigma)$ is the set of elements
of a simplicial complex $\Sigma$ partially ordered by inclusion.
Note that if $\Sigma$
is nonempty, then $\emptyset$ is the least element of $\mathcal{F}(\Sigma)$.
We warn the reader that the empty set is usually not included
in the face poset, but it will make life much easier when we
merge posets. 
Given a subset $L$ of $\Sigma^0$, the \emph{induced subcomplex} of
$\Sigma$ on $L$ is $\Sigma[L]=\{\sigma\in\Sigma | \sigma\subseteq L\}$, and
the \emph{link}
 ${\tt lk}_\Sigma(L)$ is the subcomplex $\{ \sigma\in \Sigma \mid
  \sigma\cap L =\emptyset \textrm{ and } \sigma\cup L \in \Sigma \}$  
 of $\Sigma$ with vertex set $\Sigma^0\setminus L$. 
For \emph{induced subgraphs} we use
the same notation as for induced subcomplexes.

\section{Independence complexes and discrete Morse theory}

In this section we review necessary facts regarding
discrete Morse theory and independence complexes, and
prove some useful lemmas and propositions. The
topological objects that we most often consider are
simplicial complexes,
but sometimes well-behaved finite CW-complexes pop up.
For the definition of CW-complexes and basic facts
of combinatorial topology, \cite{B} and \cite{J} are
recommended.
Discrete Morse theory is a method for reducing the number
of cells of a CW-complex without changing its homotopy
type. It was invented by Forman \cite{F1} who used the
concept of discrete Morse functions. In the last years
these functions have mostly been used only implicitly, and instead one
constructs acyclic matchings on Hasse diagrams of
face posets. In chapter 4 of Jonsson's book ``Simplicial
Complexes of Graphs'', \cite{J}, the state of art of
discrete Morse theory is surveyed. 
Our method of applying 
the theory has a lot in common with the philosophy behind 
Bousquet-M\'elou, Linusson, and Nevo's paper \cite{LNM}.

The \emph{Hasse diagram} of a poset $P$ is a directed
graph with vertex set $P$ and an arc $x\rightarrow z$
for each pair $x<z$ such that there does not exist
a $y$ satisfying $x<y<z$. The element $z$ \emph{covers}
$x$ in $P$ if $x\rightarrow z$ in the Hasse diagram.
An \emph{acyclic matching} on
$P$ is a set $\mathcal{C}$ of pairs of elements from
$P$ satisfying three conditions:
\begin{itemize}
\item[(i)] Two elements can only form a pair if one of
  them covers the other one.
\item[(ii)] No element of the poset is in more than 
            one pair of $\mathcal{C}$.
\item[(iii)] If for each pair $x\rightarrow z$
  of $\mathcal{C}$ we change the direction of the
  arcs to $x\leftarrow z$, then the Hasse diagram
  is still acylic.
\end{itemize}
We construct acyclic matchings on 
face posets, and the elements left in a poset
after the removal of all matched cells of an
acyclic matching are called the \emph{critical cells.}
Removing the cells of an acyclic matching $\mathcal{C}$
from a complex $\Sigma$ is a recurring operation and
we will use the sloppy notation $\Sigma\setminus 
\mathcal{C}$ to denoted the critical cells.
In the definition of the face poset of a 
simplicial complex we included the empty set,
if we had not done that, one of the vertices
would be a critical cell. The difference in 
the definitions corresponds
to working with either reduced $(\tilde{H}_\ast)$ or 
unreduced $(H_\ast)$ homology.

The simplical version of the main theorem of 
discrete Morse theory states that if $\Sigma$ 
is a simplicial complex and $\mathcal{C}$ is 
an acyclic matching on $\mathcal{F}(\Sigma)$ then there
is a CW-complex $\Omega$ with $\Sigma\setminus 
\mathcal{C}$ as cells (but with perhaps other
gluing maps) which is homotopy equivalent
to $\Sigma$. If no cell in the acyclic matching is
covered by a critical cell then $\Omega$ is
a simplicial complex and the homotopy equivalence
is a deformation retraction. A homological corollary from this is that 
if we have an acyclic matching $\mathcal{C}$ on $\mathcal{F}(\Sigma)$, then as vector spaces
\begin{equation}\label{eq}
 \bigoplus_i  \tilde {H}_i(\Sigma  ;\mathbb{Q}) \subseteq 
\bigoplus_{\sigma \in \Sigma \setminus \mathcal{C}  } \mathbb{Q}. 
\end{equation}

The following cluster lemma will be used to patch together acyclic matchings.

\begin{lemma}[\cite{J}, Lemma 4.2] \label{lemma:merge}
Let $\Delta$ be a simplicial complex and 
$f:\mathcal{F}(\Delta) \rightarrow P$ a poset
map to some poset $P$. If we have an acyclic
matching on each $f^{-1}(p)$ for $p\in P$,
then their union is an acyclic matching.
\end{lemma}

Our use of Lemma~\ref{lemma:merge} will follow the following pattern. For
a simplicial complex $\Sigma$ choose a subset $D$ of its vertex
set. Then consider the map $f:\mathcal{F}(\Sigma)\rightarrow
\mathcal{F}(\Sigma[D])$ defined by $\sigma \mapsto \sigma\cap D$
and use certain acyclic matchings on $f^{-1}(\tau)=\{\sigma\in\Sigma |
\sigma\cap D=\tau\}\subseteq \mathcal{F}(\Sigma)$ to obtain an
acyclic matching on all of $\mathcal{F}(\Sigma)$.

A subset $I$ of the vertex set of a graph $G$ is \emph{independent}
if there are no two vertices of $I$ that are adjacent in $G$.
The \emph{independence complex} of a graph $G$, ${\tt Ind}(G)$,
is a simplicial complex with the same vertex set as $G$ and
with faces given by the independent sets of $G$. For an introduction
to independence complexes and how discrete Morse theory can
be used on them we refer to \cite{E1,E2}.
An often used fact is that if $v$ is an isolated vertex of $G$,
then one obtains a complete acyclic matching on $\mathcal{F}({\tt Ind}(G))$
by matching each $\sigma$ which does not contain $v$ with
$\sigma\cup \{v\}$. The neighborhood $N(v)$ of a vertex
$v$ is the set of adjacent vertices.
The following is a version of the fold lemma of Engstr\"om \cite{E1,E2}.

\begin{lemma}\label{lemma:fold}
If $G$ is a graph with two distinct vertices
$u$ and $v$ which satisfy $N(u)\subseteq N(v)$,
then every acyclic matching on $\mathcal{F}({\tt Ind}(G\setminus v))$
can be extended to an acyclic matching on $\mathcal{F}({\tt Ind}(G))$
with no new critical cells.
\end{lemma}

\begin{proof}
Consider the poset map $f:\mathcal{F}({\tt Ind}(G))\rightarrow 2^{\{v\}}$ defined by 
$f(\sigma)=\sigma\cap\{v\}.$ The subposet $f^{-1}(\emptyset)$ is
$\mathcal{F}({\tt Ind}(G\setminus v))$ for which we have an acyclic matching.
Now we want an acyclic matching on $f^{-1}(\{v\})$ which is complete.
Every element $\sigma$ of $f^{-1}(\{v\})$ is an independent set which includes
$v$. Since no neighboors of $v$ are in $\sigma$, no neighboors of $u$ are in $\sigma$,
which makes $\sigma\cup\{u\}$ an independent set and an element of $f^{-1}(\{v\})$.
Clearly $\sigma\setminus \{u\} \in f^{-1}(\{v\})$ for every $\sigma \in f^{-1}(\{v\})$.
Our complete acyclic matching on $f^{-1}(\{v\})$ is then
\[ \{ (\sigma,\sigma\cup\{u\}) \mid u \not\in \sigma \in  f^{-1}(\{v\}) \}. \]
\end{proof}

The independence complex of a bunch of disjoint edges is isomorphic to the boundary of a
cross-polytope. This is the easiest non-trivial fact about independence complexes,
but we need a discrete Morse theory version of it as base case in induction
proofs later.

\begin{lemma}\label{lemma:disjointEdges}
If $G$ is the disjoint union of $n>0$ edges then there is an acyclic matching
on $\mathcal{F}({\tt Ind}(G))$ with one critical cell.
\end{lemma}

\begin{proof}
The proof is by induction on $n$. If $n=1$ and $V(G)=\{u,v\}$ then the
acyclic matching $\{(\emptyset,\{u\})\}$ has one critical cell.
If $n>1$ and $uv$ is an edge of $G$ then consider the poset map 
$f:\mathcal{F}({\tt Ind}(G))\rightarrow 2^{\{v\}}$ by 
$f(\sigma)=\sigma\cap\{v\}.$ The subposet $f^{-1}(\emptyset)$ is
$\mathcal{F}({\tt Ind}(G\setminus v))$ which has the isolated vertex
$u$ and thus gives a complete acyclic matching. From the subposet $f^{-1}(\{v\})$
there is a poset bijection to $\mathcal{F}({\tt Ind}(G\setminus\{u,v\}))$
by removing $v$, and by induction we have an acyclic matching on
$\mathcal{F}({\tt Ind}(G\setminus\{u,v\}))$ with one  critical cell.
Patching $f^{-1}(\emptyset)$ and $f^{-1}(\{v\})$ together gives one 
 critical cell.
\end{proof}

The following is a combinatorial version of the main theorem of 
Ehrenborg and Hetyei~\cite{EH} on forests.

\begin{proposition}\label{prop:Forest}
If $G$ is a forest then there is an acyclic matching on
$\mathcal{F}({\tt Ind}(G))$ with either zero or one critical cell.
\end{proposition}

\begin{proof}
We do induction on the number of edges of $G$.
If $G$ has an isolated vertex then we have an acyclic matching
with no  critical cells. If $G$ is a collection of
disjoint edges, then by Lemma~\ref{lemma:disjointEdges} there is
an acyclic matching with one  critical cell.

Otherwise there is a vertex $u$ of degree one, which is in
a connected component with more than two vertices. In that case there
has to be a vertex $v$ of distance two from $u$, and it
will satisfy $N(u)\subseteq N(v)$. By Lemma~\ref{lemma:fold}
we can extend every acyclic matching on $\mathcal{F}({\tt Ind}(G \setminus v))$
to $\mathcal{F}({\tt Ind}(G))$ without introducing new
 critical cells. And by induction there is an acyclic
mathing on $\mathcal{F}({\tt Ind}(G \setminus v))$ with none
or one  critical cells, since $G \setminus v$ is a forest.
\end{proof}

\section{Bounding Euler characteristic with the decycling number}

The following theorem is our main result.

\begin{theorem}\label{theorem:main}
\[\sum_i \dim \tilde {H}_i({\tt Ind}(G);\mathbb{Q}) \leq 
\min_{\substack{ \emptyset \neq D\subseteq V(G)\\ 
  G\setminus D \textrm{ is a forest}  } } |{\tt Ind}(G[D])|. \]
\end{theorem}
\begin{proof}
Let $D$ be a subset of $V(G)$ of size $\varphi(G)$ such that
$G\setminus D$ is a forest. If we remove even more vertices
from $G$ it will still be a forest, and so in particular, for
every $L\subseteq D,$
\[ G \setminus \Bigl( D \cup \bigcup_{v\in L} N(v) \Bigr) \]
is a forest. Now we will prove that there is an acyclic
matching on ${\tt Ind}(G)$ with at most $|{\tt Ind}(G[D])|$
 critical cells. Consider the poset map
$f: \mathcal{F}({\tt Ind}(G))\rightarrow \mathcal{F}({\tt Ind}(G[D]))$ defined by
$f(\sigma)=\sigma \cap D$. We have split the poset
into $|{\tt Ind}(G[D])|$ subposets and the next step is to show
that each of them have at most one critical cell under some
acyclic matching.
For any $L\subseteq D$ we have a poset bijection
\[  \lambda: \mathcal{F}\Bigl({\tt Ind}\Bigl(G \setminus 
    \Bigl( D \cup \bigcup_{v\in L} N(v) \Bigr)\Bigr)\Bigr)
    \rightarrow f^{-1}(L) \]
given by $\lambda(\sigma)=\sigma\cup L$. By Proposition~\ref{prop:Forest},
there is an acyclic matching on $f^{-1}(L)$ with at most
one  critical cell, since $G \setminus 
( D \cup \cup_{v\in L} N(v) )$ is a forest. By
Lemma~\ref{lemma:merge} we can patch the $|{\tt Ind}(G[D])|$
acyclic matchings together and the new acyclic matching
has at most $|{\tt Ind}(G[D])|$ critical cells.
By equality (\ref{eq}) with $\Sigma= {\tt Ind}(G)$ and
$\mathcal{C}$ as the described acyclic matchings with
$|{\tt Ind}(G[D])|$ critical cells, we are done. 
\end{proof}

The \emph{decycling number}, $\varphi(G)$, of a graph $G$ is
the minimum number of vertices whose deletion from $G$ turns
it into a forest.

\begin{corollary}\label{cor:main}
\[ |\tilde{\chi}({\tt Ind}(G))| \leq
\sum_i \dim \tilde {H}_i({\tt Ind}(G);\mathbb{Q}) \leq 
\min_{\substack{ \emptyset \neq D\subseteq V(G)\\ 
  G\setminus D \textrm{ is a forest}  } } |{\tt Ind}(G[D])|
\leq 2^{\varphi(G)}.
 \]
\end{corollary}
\begin{proof}
The left-hand inequality is
\[ |\tilde{\chi}({\tt Ind}(G))| = \Bigl| \sum_i (-1)^i \dim 
\tilde {H}_i({\tt Ind}(G);\mathbb{Q}) \Bigr|
\leq \sum_i \dim \tilde {H}_i({\tt Ind}(G);\mathbb{Q}) \]
and the right-hand inequality is
\[\min_{\substack{ D\subseteq V(G)\\ 
  G\setminus D \textrm{ is a forest}  } } |{\tt Ind}(G[D])| \leq 
  \min_{\substack{ D\subseteq V(G)\\ 
  G\setminus D \textrm{ is a forest}  } } 2^{|D|} \leq 2^{\varphi(G)}.  \]
\end{proof}

It is not hard to find examples of graphs with
$|\tilde{\chi}({\tt Ind}(G))| = 2^{\varphi(G)}$. For example if $G$ is a cycle
with $3n$ vertices, then ${\tt Ind}(G)$ is a wedge of two spheres
of the same dimension \cite{K}, and $\varphi(G)=1$. In Proposition
11.43 of \cite{J} an acyclic matching on $\mathcal{F}({\tt Ind}(G))$ 
with two critical cells is constructed.

\section{Bounds for some lattices}\label{sec:bounds}

Recall that the Fibonacci number $F_n$ is defined by $F_1=F_2=1$ and
$F_n=F_{n-1}+F_{n-2}$ for $n>2$, and the sequence starts with $1,1,2,3,5,8.$
Explicitly we have $F_n = 5^{-1/2}(\phi^n-(-\phi)^{-n})$ where $\phi$ is
the golden ratio $(1+\sqrt{5})/2$.
The graph $P_n$ is the path on $n$ vertices.

\begin{proposition}
$|{\tt Ind}(P_n)|=F_{n+2}.$
\end{proposition}
\begin{proof}
Clearly $|{\tt Ind}(P_1)|=2=F_3$ and $|{\tt Ind}(P_2)|=3=F_4$. Let $n>2$.
If the last vertex of the path is occupied, the one next to it is empty,
and the other ones can be picked in $|{\tt Ind}(P_{n-2})|$ ways. If it
is not occupied, the rest can be picked in $|{\tt Ind}(P_{n-1})|$ ways.
\end{proof}

Now we will use the results from the previous section on some lattices. In 
each figure there are three lattices illustrated, and from left to right they are: 
The lattice we want to calculate the Witten index for, the acyclic lattice, 
and the lattice of removed vertices. For large lattices the influence
from the choice of open, cylindrical, or closed boundaries is negligible.

\subsubsection*{The hexagonal lattice}
\begin{center}
  \includegraphics*{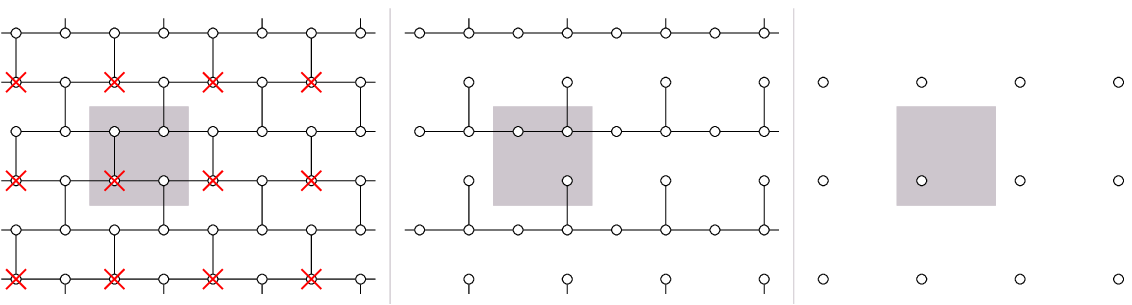}
\end{center}
From a $2m\times 2n$ hexagonal lattice we remove $mn$ vertices to get an
acyclic lattice. By Corollary~\ref{cor:main}, the absolute value of
the Witten index is at most $2^{mn}$ which is $2^{1/4}\approx 1.19$
per vertex.

\subsubsection*{The hexagonal dimer lattice}
\begin{center}
  \includegraphics*{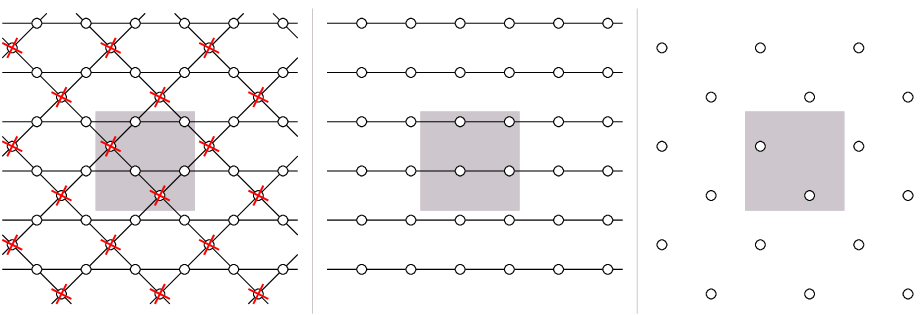}
\end{center}
A hexagonal dimer lattice built from $m\times n$ grey blocks has $6mn$
vertices and we remove $2mn$ of them to get an acyclic lattice, and
so by Corollary~\ref{cor:main}, $|W|\leq 2^{2mn}$ which is 
$2^{1/3}\approx 1.26$ per vertex.

\subsubsection*{The triangular lattice}
\begin{center}
  \includegraphics*{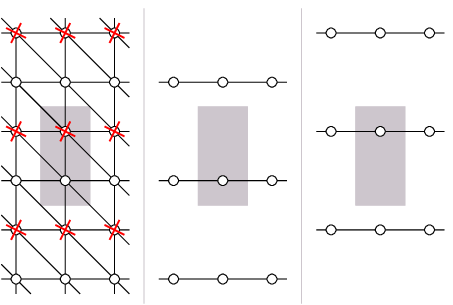}
\end{center}
From a triangular $2m\times n$ lattice we remove $m$ paths of length
$n$ to get an acyclic lattice. By Corollary~\ref{cor:main} we have
that 
\[|W|\leq |{\tt Ind}(P_n)|^m = F_{n+2}^m = 
  5^{-m/2}(\phi^{n+2}-(-\phi)^{-n-2})^m
   \approx \phi^{mn}\]
with an approximate $\sqrt{\phi}\approx 1.27$ contribution per vertex.

\subsubsection*{The triangular dimer lattice}
\begin{center}
  \includegraphics*{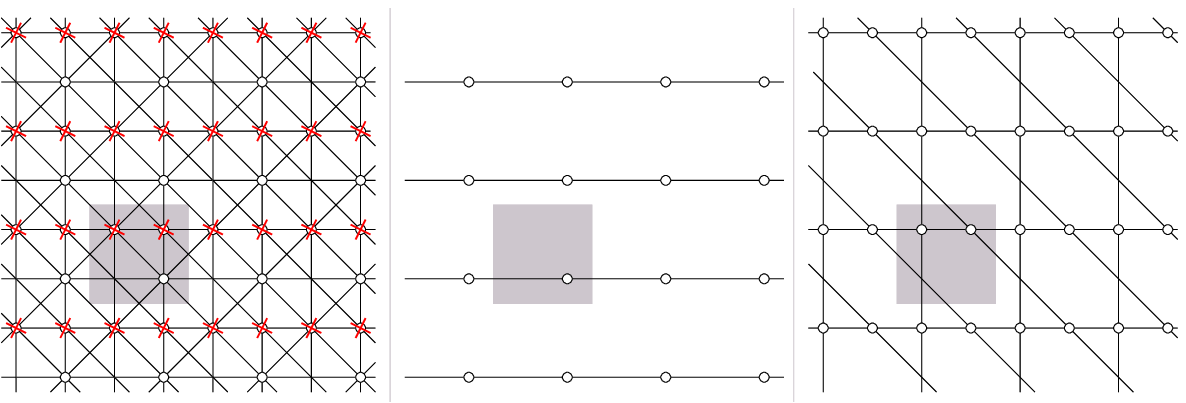}
\end{center}
A triangular dimer lattice built from $m\times n$ grey blocks has $3mn$
vertices
and we remove $2mn$ of them to get an acyclic lattice. The vertices
we removed induce a lattice for which the size of the independence
complex is not easily calculated. If we remove edges from it we get
more independent sets and a weaker upper bound, but perhaps
a computable one. Remove all diagonal edges to get $m$ paths of 
length $2n$, and by Corollary~\ref{cor:main} we have,
\[|W|\leq |{\tt Ind}(P_{2n})|^m = F_{2n+2}^m = 
  5^{-m/2}(\phi^{2n+2}-(-\phi)^{-2n-2})^m
   \approx \phi^{2mn}\]
with an approximate $\phi^{2/3} \approx 1.38$ contribution per vertex.

\subsubsection*{The square dimer lattice}
\begin{center}
  \includegraphics*{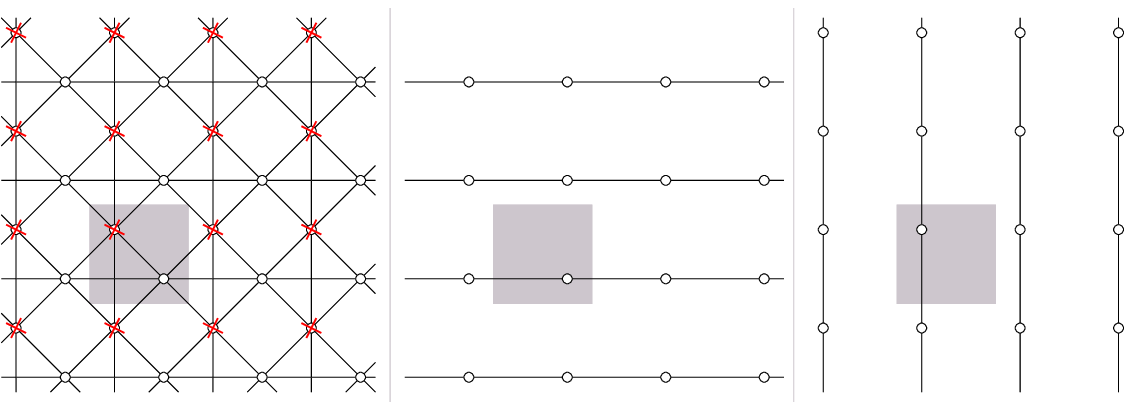}
\end{center}
From a square dimer $2m\times 2n$ lattice we remove $n$ paths of length
$m$ to get an acyclic lattice. By Corollary~\ref{cor:main} we have
that 
\[|W|\leq |{\tt Ind}(P_m)|^n = F_{m+2}^n = 
  5^{-n/2}(\phi^{m+2}-(-\phi)^{-m-2})^n
   \approx \phi^{mn}\]
with an approximate $\sqrt{\phi}\approx 1.27$ contribution per vertex.

\subsection{A comparison with van Eertens calculations}

Using computer calculations for lattices of size $m\times n$ 
with $m,n\leq 15$, van Eerten \cite{vE} approximated the contribution
to $|W|$ per vertex.

\begin{center}
\begin{tabular}{l|ll}
Lattice type & van Eertens value & Upper bound in this paper \\
\hline
Hexagonal & $1.2$ & $1.19$ \\
Hexagonal dimer & $1.25$ & $1.26$ \\
Triangular & $1.14$ & $1.27$ \\
Triangular dimer & $1.36$ & $1.38$ \\
Square dimer & $1.15$ & 1.27 \\
\end{tabular}
\end{center}

The values from \cite{vE} for dimer-models are tabulated here
per vertex and not per site.

\section{The cohomological method and the 3-rule}

In this section we treat the case that $G\setminus D$ is
not only a forest, but it completely lacks edges. If we
also impose conditions on the differentials
of the Morse complex \cite{F1}, then we recover
the \emph{cohomological method} of \cite{FHHS, FS, HS}.

\begin{theorem}\label{theo:coho}
If $G$ is a graph and $D$ a set of vertices such that
$G\setminus D$ has no edges, then there is a Morse matching
on ${\tt Ind}(G)$ whose critical cells are the 
$\sigma\in {\tt Ind}(G[D])$ such that 
\[ \bigcup_{v\in \sigma}N(v) \supseteq V(G)\setminus D. \]
\end{theorem}

\begin{proof}
Use the same Morse matching as in the proof of 
Theorem \ref{theorem:main}.
\end{proof}

Two vertices of a graph are at least distance
three apart if they are non-adjacent and share no neighbors.
The following is a generalization of the \emph{``3-rule''} of
\cite{FHHS, FS, HS}.

\begin{corollary}\label{cor:3}
If $G$ is a graph and $R$ is a set of vertices such that
\begin{itemize}
  \item[i)] all pairs of vertices of $R$ are
     at least distance three apart, and
  \item[ii)] no independent set of $G\setminus R$ is larger
     than $R$,
\end{itemize}
then ${\tt Ind}(G)$ is a wedge of $K$ spheres of 
dimension $(|R|-1)$. Construct a graph $G'$ by starting with 
$G\setminus R$ and add cliques on $\{ v \mid vw \textrm{ edge of } G\}$ for all $w\in R$. 
The number of independent sets of $G'$ with $|R|$ elements is $K$.
\end{corollary}

\begin{proof}
The set $D$ in Theorem \ref{theo:coho} is $V(G)\setminus R$. 
For any $v\in V(G)\setminus R$ its neighborhood can only contain one vertex in $R$, 
since the vertices in $R$ are pairwise at least distance three apart.
So to get a $\sigma\in {\tt Ind}(G\setminus R)$ such that 
$ \cup_{v\in \sigma}N(v) \supseteq R $ we need a $\sigma$
with at least $|R|$ elements. But that is also the
maximum size of an independent set of $G\setminus R$.

The property that $ \cup_{v\in \sigma}N(v) \supseteq R $
for some $\sigma$, can now be restated as: for every $w\in R$
there is a unique $v\in\sigma$ such that $w\in N(v)$. Enforcing
this condition on the maximal independent sets of $G\setminus R$ is the same
as adding cliques on $N(w)$ for all $w\in R$.

Since all critical cells of the matching are of the same dimension 
 ${\tt Ind}(G)$ is a wedge of spheres.
\end{proof}

In \cite{FHHS, FS, HS} it is described, in the context of the
cohomological method, how the generators of cohomology of 
${\tt Ind}(G)$ are related to the ground states of the
supersymmetric model on $G$. When ${\tt Ind}(G)$ is 
isomorphic to a wedge of spheres of the same dimension then 
the number of ground states is the number of spheres. 

The two standard examples of the use of the cohomological method
and the 3-rule are the cycle with $3n$ vertices and the martini lattice.
For a cycle on the $3n$ vertices $0,1,2,\ldots,3n-1$ with edges $(v,v+1)$,
let $R=\{0,3,6,\ldots,3n-3\}$. The graph $G'$ of Corollary~\ref{cor:3}
is a cycle on $2n$ vertices and the ground states are represented by
the independent sets of $G'$ on $n$ vertices. There are two of them.

The martini lattice is not new, but we present it as a first example
of a general procedure to obtain lattices that satisfy the conditions
of Corollary \ref{cor:3}. First we pick a regular bipartite graph, a
hexagonal lattice with closed boundaries.
\begin{figure}
\begin{center}
  \includegraphics*{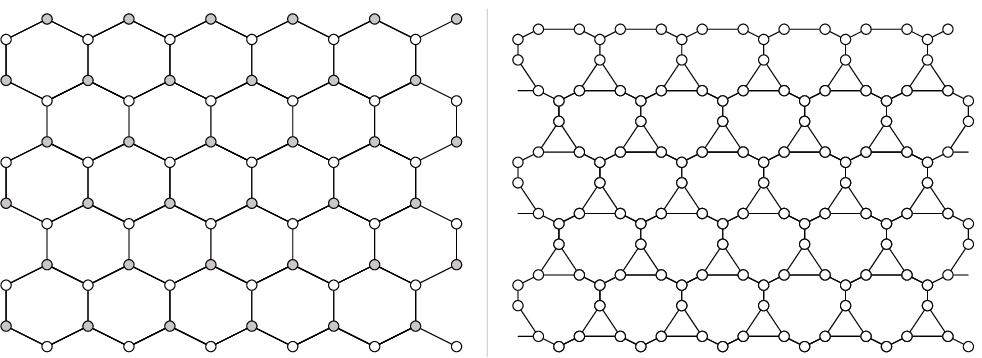}
\begin{caption}{The hexagonal and martini lattices.}\label{fig:hexToMar} 
\end{caption}
\end{center}
\end{figure}
The bipartition is indicated by white and grey vertices in Figure~\ref{fig:hexToMar}.
Transform the grey vertices from Y to $\Delta$ as in Figure~\ref{fig:hexToMar} to get the martini lattice.
The untransformed vertices form the set $R$.
\begin{figure}
\begin{center}
  \includegraphics*{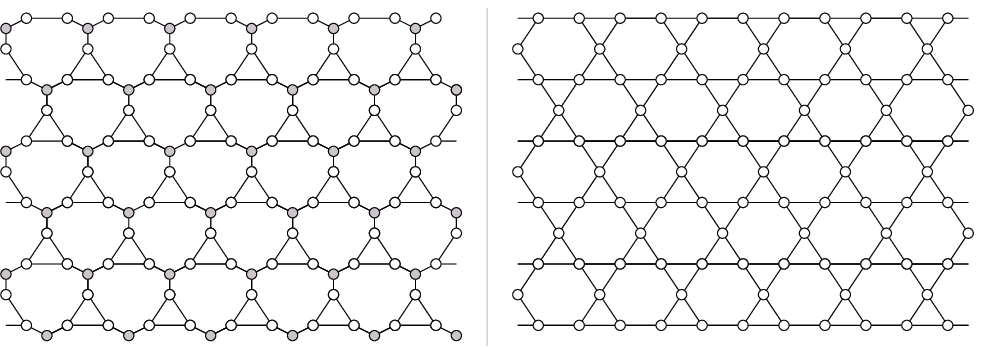}
\begin{caption}{The martini lattice with $R$ in grey. The graph $G'$ is the hexagonal dimer lattice.}\label{fig:marToHexDimer} 
\end{caption}
\end{center}
\end{figure}
Replace the vertices of $R$
with cliques to get $G'$ as in Figure~\ref{fig:marToHexDimer}. By Corollary \ref{cor:3} the
maximal independent sets of $G'$ in Figure~\ref{fig:marToHexDimer} counts the ground
states of the martini lattice. Comparing the hexagonal lattice in Figure~\ref{fig:hexToMar}
with $G'$ in Figure~\ref{fig:marToHexDimer} one notices that $G'$ is the hexagonal
dimer lattice. Ending up with the dimer lattice is a general feature of the procedure
examplified on the hexagonal lattice. Counting maximal independent sets of the
hexagonal dimer lattice is the same as counting perfect matchings on the hexagonal
lattice, and that is solved \cite{Ka,W}.

Now we repeat the same procedure but start off with a 3D-grid with closed boundaries.
\begin{figure}
\begin{center}
  \includegraphics*{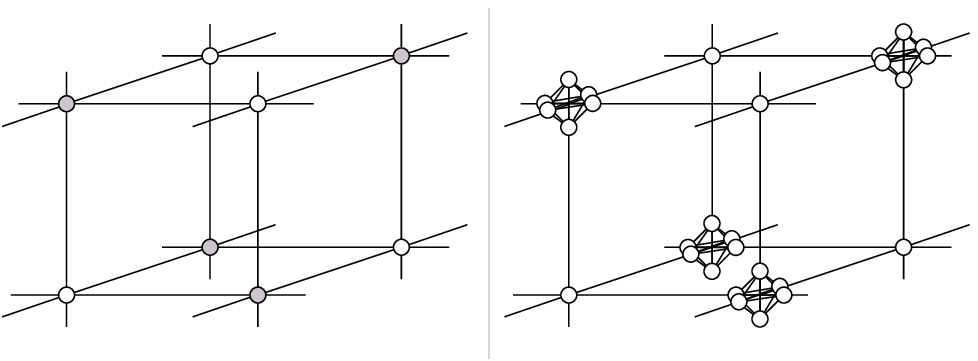}
\begin{caption}{The 3D-grid and the semi-dimer 3D-grid.}\label{fig:3D} 
\end{caption}
\end{center}
\end{figure}
A piece of the 3D-grid with a bipartition into white and grey vertices is drawn in
Figure~\ref{fig:3D}. Replace every grey vertex with a complete graph of the same
order as the vertex degree to get the \emph{semi-dimer 3D-grid} in Figure~\ref{fig:3D}.
By Corollary~\ref{cor:3} the number of ground states for the semi-dimer 3D-grid 
is the same as the number of perfect matchings on the 3D-grid, and there
are good bounds for those as well \cite{C,S}.

For any lattice obtained from this procedure there is a good lower bound
on the number of ground states. It follows from Schrijver's \cite{S} result that
there are at least
\[ \left( \frac{(k-1)^{k-1}}{k^{k-2}} \right) ^ {n/2} \]
perfect matchings on a $k$-regular bipartite graph on $n$ vertices. 
\begin{figure}
\begin{center}
  \includegraphics*{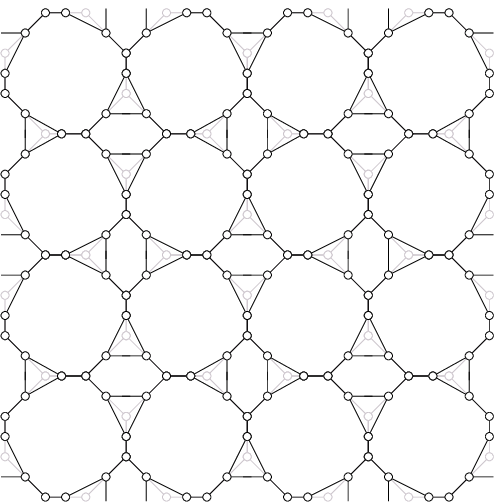}\rule{0.3cm}{0pt}\includegraphics*{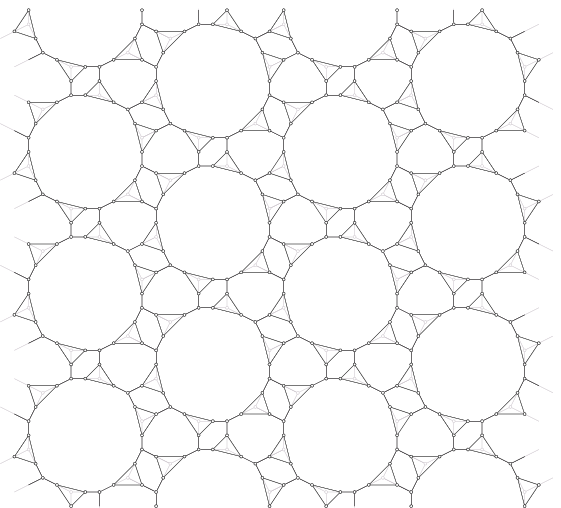}
\begin{caption}{Graphs produced from the edges of the truncated square tiling and the great rhombitrihexagonal tiling \cite{GS}.}\label{fig:hypo} 
\end{caption}
\end{center}
\end{figure}
In Figure~\ref{fig:hypo} two graphs produced from 3-regular bipartite graphs
are illustrated. According to Schrijver's bound there are at least $(4/3)^{n/2}$
perfect matchings.
We can now
construct lattices in arbitrary dimensions with more than $\alpha^n$ ground states
for $\alpha >1$ according to the following construction. For $d>1$ dimensions 
let $n_1,n_2,\ldots,n_d$ be even positive
numbers larger than two. Start with the $2d$--regular 
bipartite graph $T=C_{n_1}\times C_{n_2} \times \cdots \times C_{n_d}$,
a $d$-dimensional grid with closed boundaries. 
Take one of the two parts of $T$ and replace each vertex with $2d$ new vertices
as for the 3D-grid. Then we get a $d$-dimensional lattice with at least
\[ \left( \frac{(2d-1)^{2d-1}}{(2d)^{2d-2}} \right) ^ {\frac{n_1 n_2 \cdots n_d}{2} } \]
ground states.

\subsubsection*{Acknowledgements}
The author thanks Christian Krattenthaler, Philippe Di Francesco,
and the Mathematisches Forschungsinstitut Oberwolfach for organizing a week
on enumerative combinatorics and statistical mechanics; 
Anton Dochtermann and Jakob Jonsson for their comments on the paper;
and the referees for their suggestions, in particular regarding the
connections to the cohomological method and the 3-rule.

\bibliographystyle{amsplain}

\end{document}